\documentclass[11pt]{article}

\usepackage{enumitem}
\usepackage{titling}

\usepackage{palatino} 
\usepackage{mathpazo}
\usepackage{braket}
\usepackage{amsfonts}
\usepackage{amssymb}
\usepackage{amsmath}
\usepackage{latexsym}
\usepackage{amsthm}
\usepackage[usenames]{color}
\usepackage{hyperref}

\usepackage{tabularx}

\hypersetup{pdfpagemode=UseNone}

\newtheorem{theorem}{Theorem}%[section]
\newtheorem{lemma}{Lemma}

\newtheorem{prop}{Proposition}
\newtheorem{corollary}{Corollary}

\newtheorem{protocol}{Protocol}
\newtheorem{remark}{Remark}

\newcommand{\tr}{\operatorname{Tr}}
\newcommand{\chews}[2]{\scriptsize{\left( \!\!\! \begin{array}{c} #1 \\ #2 \end{array} \!\!\! \right)}}
\newcommand{\eps}{\varepsilon}
\newcommand{\inner}[2]{\langle #1, #2 \rangle}

\newcommand{\calA}{\mathcal{A}}
\newcommand{\calB}{\mathcal{B}}

\newcommand{\C}{\mathbb{C}}

\newcommand{\spann}{\mathrm{span}}

\newcommand{\DRIC}{\mathrm{DRIC}}

\newcommand{\dsum}{\displaystyle\sum}

\newcommand{\ketbra}[2]{\ket{#1} \bra{#2}}
\newcommand{\kb}[1]{\ketbra{#1}{#1}}
\newcommand{\brakett}[2]{\langle #1| #2 \rangle}

\def\modd{\textup{ mod }}

\begin{document}

%-----------------------------------------------------------------------------%
\title{\bf 
Simple, near-optimal quantum protocols for die-rolling
}
%-----------------------------------------------------------------------------%

\author{ 
Jamie Sikora\thanks{Centre for Quantum Technologies, National University of Singapore, and MajuLab, CNRS-UNS-NUS-NTU International Joint Research Unit, UMI 3654, Singapore. Email \tt{cqtjwjs@nus.edu.sg}}
}

\date{August 28, 2018}

\maketitle

\begin{abstract}
\emph{Die-rolling} is the cryptographic task where two mistrustful, remote parties wish to generate a random 
$D$-sided die-roll over a communication channel. 
Optimal quantum protocols for this task have been given by Aharon and Silman (New Journal of Physics, 2010) but are based on optimal weak coin-flipping protocols which are currently very complicated and not very well understood. 
In this paper, we first present very simple classical protocols for die-rolling
which have decent (and sometimes optimal) security which is in stark contrast to coin-flipping, bit-commitment, oblivious transfer, and many other two-party cryptographic primitives. 
We also present quantum protocols based on integer-commitment, a generalization of bit-commitment, where one wishes to commit to an integer. 
We analyze these protocols using semidefinite programming and finally give protocols which are very close to Kitaev's lower bound for any $D \geq 3$. 
Lastly, we briefly discuss an application of this work to the quantum state discrimination problem. 
\end{abstract} 

\section{Introduction} 
 
\emph{Die-rolling} is the two-party cryptographic primitive in which two spatially separated parties, Alice and Bob, wish to agree upon an integer  
$d \in [D] := \{ 1, \ldots, D \}$, generated uniformly at random, 
over a communication channel. 
When designing die-rolling protocols, the security goals are: 
\begin{enumerate}
  \item \emph{Completeness:} If both parties are honest, then their outcomes are the same, uniformly random, and neither party aborts. 
  \item \emph{Soundness against cheating Bob:} If Alice is honest, then a dishonest (i.e., cheating) Bob cannot influence her protocol outcome away from uniform. 
  \item \emph{Soundness against cheating Alice:} \label{it:security:cheating-Alice} If Bob is honest, then a dishonest (i.e., cheating) Alice cannot influence his protocol outcome away from uniform.  
\end{enumerate}

We note here that Alice and Bob start uncorrelated and unentangled. 
Otherwise, Alice and Bob could each start with half of the following maximally entangled state 
\[ 
\frac{1}{\sqrt D} \sum_{d \in [D]} \ket{d}_{\calA} \ket{d}_{\calB} 
\] 
and measure in the computational basis to obtain a perfectly correlated, uniformly random die-roll. 
Thus, such a primitive would be trivial if they were allowed to start entangled.   

Die-rolling is a generalization of a well-studied primitive known as \emph{coin-flipping}~\cite{Blu81} which is the special case of die-rolling when $D = 2$. 
In this paper, we analyze die-rolling protocols in a similar fashion that is widely adopted for coin-flipping protocols~\cite{ATVY00, NS03, KN04, Moc07, CK09, NST15, NST16}. 
That is, we assume perfect completeness and calculate the soundness in terms of the \emph{cheating probabilities}, as defined by the symbols:   
\begin{center}
\begin{tabularx}{\textwidth}{rX}
  $P_{B,d}^*$: & The maximum probability with which a dishonest Bob can force an honest Alice to accept the outcome $d \in [D]$ by digressing from protocol. 
  \\
  $P_{A,d}^*$: & The maximum probability with which dishonest Alice can force an honest Bob to accept the outcome $d \in [D]$ by digressing from protocol. 
\end{tabularx} 
\end{center} 
We are concerned with designing protocols which minimize the maximum of these $2D$ quantities since a protocol is only as good as its worst cheating probability. 
Coincidentally, all the protocols we consider in this paper have the property that all of Alice's cheating probabilities are equal and similarly for a cheating Bob. 
Therefore, for brevity, we introduce the following shorthand notation: 
\[ 
P_A^* := \max \{ P_{A,1}^*, \ldots, P_{A,D}^* \} 
\quad 
\text{ and } 
\quad 
P_B^* := \max \{ P_{B,1}^*, \ldots, P_{B,D}^* \}.  
\] 

When $D=2$, the security definition for die-rolling above aligns with that of \emph{strong} coin-flipping. 
For strong coin-flipping, it was shown by Kitaev~\cite{Kit03} that any quantum protocol satisfies $P_{A,1}^* P_{B,1}^* \geq 1/2$ and $P_{A,2}^* P_{B,2}^* \geq 1/2$, implying that at least one party can cheat with probability at least $1/\sqrt{2}$. 
It was later shown by Chailloux and Kerenidis~\cite{CK09} that all four cheating probabilities can be made arbitrarily close to $1/\sqrt{2}$ by using optimal quantum protocols for \emph{weak} coin-flipping as discovered by Mochon~\cite{Moc07}. 

Kitaev's proof for the lower bound on coin-flipping extends naturally to die-rolling; it can be shown that for any quantum die-rolling protocol, we have 
\[ P_{A,d}^* P_{B,d}^* \geq \frac{1}{D} \] 
for any $d \in [D]$. 
This implies the lower bound $\max \{ P_{A}^*, P_{B}^* \} \geq 1/\sqrt{D}$. 
In fact, extending the optimal coin-flipping protocol construction in~\cite{CK09}, it was shown by Aharon and Silman~\cite{AS10} that for $D > 2$, it is possible to find quantum protocols where the maximum of the $2D$ probabilities is at most $1/\sqrt{D} + \delta$, for any $\delta > 0$.  
  
The optimal protocols in~\cite{CK09} and \cite{AS10} are not explicit as they rely on using Mochon's optimal weak coin-flipping protocols as subroutines. 
Moreover, Mochon's protocols are very complicated and not given explicitly, although they have been simplified~\cite{ACGKM15}. 

The best known \emph{explicit} quantum protocol for die-rolling\footnote{The protocols considered in this paper have a much different form than these protocols.} of which we are aware is given in~\cite{AS10}. It uses three messages and has cheating probabilities 
\[  
P_A^* := \frac{D+1}{2D}
\quad \text{ and } \quad 
P_B^* := \frac{2D-1}{D^2}. 
\] 
These probabilities have the attractive property of approximating Kitaev's lower bound in the limit, but since $P_A^* \to 1/2$ as $D \to \infty$, the maximum cheating probability is quite large. 

This motivates the work in this paper which is to find simple and explicit protocols for die-rolling that approximate Kitaev's lower bound on the maximum cheating probability \[ \max \{ P_A^*, P_B^* \} \geq 1/\sqrt{D}. \] 
 
\subsection{Simple classical protocols}

We first show that simple classical protocols exist with decent security. 

\begin{protocol}[Classical protocol] 
\label{protocol:classical}
\quad 
\begin{itemize}[noitemsep] 
\item Bob chooses a subset $S \subseteq [D]$ with $|S| = m$, uniformly at random, and sends $S$ to Alice. If $|S| \neq m$, Alice aborts. 
\item Alice selects $d \in S$ uniformly at random and tells Bob her selection. If $d \not\in S$, Bob aborts.
\item Both parties output $d$. 
\end{itemize} 
\end{protocol}  

We see that this is a valid die-rolling protocol as each party outputs the same value $d \in [D]$ and each value occurs with equal probability. 
As for the cheating probabilities, it is straightforward to see that 
\[ 
P_A^* = \frac{m}{D}  
\quad \text{ and } \quad 
P_B^* = \frac{1}{m}.  
\]
Besides being extremely simple, this protocol has the following interesting properties: 
\begin{itemize}[noitemsep] 
\item The product $P_{A,d}^* P_{B,d}^* = 1/D$, for any $d \in [D]$,  saturates Kitaev's lower bound for every $d \in [D]$.  
\item For $D$ square and $m = \sqrt{D}$, we have $P_A^* = P_B^* = 1/\sqrt{D}$, yielding an optimal protocol! 
\item If $D$ is not square, the protocol is not \emph{fair},  meaning that $P_A^* \neq P_B^*$.  
\end{itemize} 

Note that to minimize $\max \{ P_A^*, P_B^* \}$, it does not make sense to choose large $m$ (greater than $\lceil \sqrt{D} \rceil$) or small $m$ (less than $\lfloor \sqrt{D} \rfloor$). 
We can see that for $D = 3$, $D = 7$, or $D = 8$, for example, that choosing the ceiling is better while for $D = 5$ or $D = 10$ choosing the floor is better. 
Thus, we keep both the cases and summarize the overall security of the above protocol in the following lemma. 
 
\begin{lemma} 
\label{lem:classical}
For $D \geq 2$, there exists a classical die-rolling protocol satisfying 
\begin{equation} 
\label{eqn:classical_lemma}
\dfrac{1}{\sqrt D} \leq \max \{ P_A^*, P_B^* \} 
= 
\min \left\{ \dfrac{\lceil \sqrt{D} \rceil}{D}, \dfrac{1}{\lfloor \sqrt{D} \rfloor} \right\} 
\end{equation}
which is optimal when $D$ is square. 
\end{lemma} 
 
Note that the special case of $D = 2$ has either Alice or Bob able to cheat perfectly, which is the case for all classical coin-flipping protocols. 
However, Kitaev's bound on the product of cheating probabilities is still (trivially) satisfied. 
For $D = 3$, we can choose $m = 2$ to obtain $\max \{ P_A^*, P_B^* \} = 2/3$ proving that even classical protocols can have nontrivial security, which is vastly different than the $D = 2$ case. 
The values from (\ref{eqn:classical_lemma}) for $D \in \{ 2, \ldots, 10 \}$ are later presented in Table~\ref{table:corr_vals}. 
  
\begin{remark} 
The protocol above can easily be extended for more parties. 
For example, if $k$ parties wish to roll $k$ $D$-sided dice. 
If each party rolls one die each, then if $k-1$ parties cheat, they cannot force a specific $k$-dice outcome with probability greater than $1/D$ (the probability the honest die landed in their favour). 
This is optimal in the $k$-party $D^k$-outcome setting. 
(See the work of Aharon and Silman~\cite{AS10} for further details about this setting and its security definitions.)    
\end{remark}   
  
We are not aware of other lower bounds for classical die-rolling protocols apart from those implied by Kitaev's bounds above. 
We see that sometimes classical protocols can be optimal, for example when $D$ is square. 
We now consider how to design (simple) quantum protocols and see what levels of security they can offer. 
 
\subsection{Quantum protocols based on integer-commitment} 
 
Many of the best known explicit protocols for strong coin-flipping are based on \emph{bit-commitment} \cite{Amb01,SR01,KN04, NST16}. 
Optimal protocols are known for bit-commitment as well~\cite{CK11}, but are again based on weak coin-flipping and are thus very complicated. 
  
In this paper, we generalize the simple, explicit protocols based on bit-commitment such that Alice commits to an \emph{integer} instead of a bit. 
More precisely, a quantum protocol based on integer-commitment has the following form.   

\begin{protocol}[Quantum protocol] 
\label{prot:DRDC} 
A quantum die-rolling protocol based on integer-commitment, denoted here as $\DRIC$, is defined as follows: 
\begin{itemize}[noitemsep, nolistsep]
\item Alice chooses a random $a \in [D]$ and creates the state $\ket{\psi_a} \in \calA \otimes \calB$ and sends the subsystem $\calB$ to Bob. 
\item Bob sends a uniformly random $b \in [D]$ to Alice. 
\item Alice reveals $a$ to Bob and sends him the subsystem $\calA$. 
\item Bob checks if $\calA \otimes \calB$ is in state $\ket{\psi_a}$ using the measurement $\{ \Pi_{a} := \kb{\psi_a}, \, \Pi_{abort} := I -\Pi_{a} \}$.  
Bob accepts/rejects $a$ based on his measurement outcome. 
\item If Bob does not abort, Alice and Bob output $d := a + b \modd D + 1 \in [D]$. 
\end{itemize} 
\end{protocol}  

The special case of $D = 2$ yields the structure of the simple, explicit coin-flipping protocols mentioned above. 
Indeed, these protocols are very easy to describe, one needs only the knowledge of the $D$ states $\ket{\psi_a}$ and, implicitly, the systems they act on, $\calA$ and $\calB$. 
 
We start by formulating the cheating probabilities of a $\DRIC$-protocol using semidefinite programming. Once we have established the semidefinite programming cheating strategy formulations, we are able to analyze the security of $\DRIC$-protocols. 
Furthermore, we are able to analyze \emph{modifications} to such protocols and the corresponding changes in security. 

In this paper, we present a $\DRIC$-protocol with near-optimal security. 
We develop this protocol in several steps described below. 
 
The first step is to start with a protocol with decent security. 
To do this, we show how to create a $\DRIC$-protocol with the same cheating probabilities as Protocol~\ref{protocol:classical}. 

\begin{prop} \label{prop:DRICprotocol} 
There exists a $\DRIC$-protocol with the same cheating probabilities as in Protocol~\ref{protocol:classical}. 
\end{prop}  

The second step is to give a process where \emph{any} $\DRIC$-protocol can be made more fair, i.e., to try to equate the maximum cheating probabilities of Alice and Bob. 
We accomplish this by modifying the protocol in order to balance Alice and Bob's cheating probabilities, thereby decreasing the overall maximum cheating probability. 

\begin{prop} \label{prop:balancing} 
If there exists a $\DRIC$-protocol with cheating probabilities $P_A^* = \alpha$ and $P_B^* = \beta$, then there exists a $\DRIC$-protocol with maximum cheating probability 
\[ \max \{ P_A^*, P_B^* \} \leq \frac{D \max \{ \beta, \alpha \} - \min \{ \beta, \alpha \}}{D |\beta -  \alpha| + D - 1} \leq \max \{ \beta, \alpha \}. \] 
Moreover, the last inequality is strict when $\alpha \neq \beta$ yielding a strictly better protocol.  
\end{prop}  
 
By combining the above two propositions, we are able to obtain the main result of this paper. 
  
\begin{theorem} 
\label{thm:main}
For any $D \geq 2$, there exists a (quantum) $\DRIC$-protocol  satisfying 
\[ 
\frac{1}{\sqrt D} \leq \max\{ P_{A}^*, P_{B}^* \} \leq \min \left\{ 
\dfrac{D + \lfloor \sqrt{D} \rfloor}{D(\lfloor \sqrt{D} \rfloor + 1)}
,
\dfrac{1 + \lceil \sqrt{D} \rceil}{D + \lceil \sqrt{D} \rceil}
\right\}  
\] 
which is strictly better than Protocol~\ref{protocol:classical} when $D$ is not square. 
\end{theorem} 

Since 
$\min \left\{ \dfrac{D + \lfloor \sqrt{D} \rfloor}{D(\lfloor \sqrt{D} \rfloor + 1)}, \dfrac{1 + \lceil \sqrt{D} \rceil}{D + \lceil \sqrt{D} \rceil} \right\} \approx \dfrac{1}{\sqrt D}$ for large $D$, this bound is very close to optimal. 
To compare numbers, we list the values for $D \in \{ 2, \ldots, 10 \}$, below. 

\begin{table}[ht]
\caption{
Values of our bounds (as truncated percentages) for various protocols and values of $D$. 
We see that the quantum protocol performs very well, even for $D$ as small as $3$.
} 
\label{table:corr_vals} 
\begin{center}
\begin{tabular}{|r||c|c|c|c|c|c|c|c|c|}
\hline
$D$ & $2$ & $3$ & $4$ & $5$ & $6$ & $7$ & $8$ & $9$ & $10$ \\  
\hline 
Explicit Protocol in~\cite{AS10} & $75 \%$ & $66 \%$ & $62 \%$ & $60 \%$ & $58 \%$ & $57 \%$ & $56 \%$ & $55 \%$ & $55 \%$ \\
\hline 
Our Classical Protocol & $100 \%$ & $66 \%$ & $50 \%$ & $50 \%$ & $50 \%$ & $42 \%$ & $37 \%$ & $33 \%$ & $33 \%$ \\
\hline
\textbf{Our Quantum Protocol} & $75 \%$ & $60 \%$ & $50 \%$ & ${46 \%}$ & $44 \%$ & $40 \%$ & $36 \%$ & $33 \%$ & ${32 \%}$ \\ 
\hline
Kitaev's lower bound & $70 \%$ & $57 \%$ & $50 \%$ & $44 \%$ & $40 \%$ & $37 \%$ & $35 \%$ & $33 \%$ & $31 \%$ \\ 
\hline 
\end{tabular}
\end{center}
\end{table} 

\vspace{-0.25cm}
 
\paragraph{Related literature.} 
Quantum protocols for a closely related cryptographic task known as string-commitment have been considered~\cite{Kent03, Tsu05, Tsu06, BCHLW08, Jain08}. 
Technically, this is the case of integer-commitment when $D = 2^n$ (if the string has $n$ bits). 
It is worth noting that the quantum protocols considered in this paper are quite similar, but the security definitions are very different. 
Roughly speaking, they are concerned with quantum protocols where Alice is able to ``cheat'' on $a$ bits and Bob is able to ``learn'' $b$ bits of information about the $n$ bit string. Multiple protocols and security trade-offs are given in the above references. 

The use of semidefinite programming has been very valuable in the study of quantum cryptographic protocols, see for example~\cite{Kit03, Moc05, Moc07, CKS13, NST15, NST16}. 
Roughly speaking, if one is able to formulate cheating probabilities as semidefinite programs, then the problem of analyzing cryptographic security can be translated into a concrete mathematical problem. 
Moreover, one then has the entire theory of semidefinite programming at their disposal. 
This is the approach taken in this work, to shine new light on a cryptographic task using the lens of semidefinite programming. 
  
\subsection{Kitaev's lower bound and the quantum state discrimination problem} 

The security analysis of $\DRIC$-protocols has many similarities to the \emph{quantum state discrimination problem}. 
Suppose you are given a quantum state $\rho \in \{ \rho_1, \ldots, \rho_n \}$ with respective  probabilities $p_1, \ldots, p_n$. 
The quantum state discrimination problem is to determine which state you have been given (by means of measuring it) with the maximum probability of being correct. We only briefly discuss this problem in this work; the interested reader is referred to the survey~\cite{Spehner14} and the references therein. 

We give a very short proof of Kitaev's lower bound for the special case of $\DRIC$-protocols. Afterwards, we show that it can be generalized to show the following bound for the quantum state discrimination problem. 

\begin{prop} \label{prop:QSD}
If given a state from the set $\{ \rho_1, \ldots, \rho_n \}$, with respective probabilities $\{ p_1, \ldots, p_n \}$, then there exists a POVM to learn which state was given with success probability at least 
$\lambda_{\min} \left( \left( \sum_{i=1}^n W_i^{-1} \right)^{-1} \right)$ 
for any positive definite Hermitian 
$\{ W_1, \ldots, W_n \}$ satisfying $\inner{W_i}{\rho_i} \leq 1$, for all $i \in [n]$. 
Here, $\lambda_{\min}$ denotes the smallest eigenvalue of a Hermitian matrix. 
\end{prop} 

Note that the above proposition is indeed independent of the $p_i$'s and could thus probably be strengthened. However, we use cryptographic reasoning to argue that  this bound can be tight.     
   
\paragraph{Paper organization.} 

In Section~\ref{sect:SDP} we develop the semidefinite programming cheating strategy formulations for Alice and Bob. In Section~\ref{sect:qprotocol} we exhibit a $\DRIC$-protocol then use the semidefinite programming formulations to prove Proposition~\ref{prop:DRICprotocol}, that the protocol has the same cheating probabilities as Protocol~\ref{protocol:classical}. Section~\ref{sect:balancing} shows how to balance the probabilities in a $\DRIC$-protocol by showing how to reduce Bob's cheating, then how to reduce Alice's. Combining these yields a proof of Proposition~\ref{prop:balancing}. Lastly, in Section~\ref{sect:Kitaev}, we give a short proof of Kitaev's lower bound when applied to $\DRIC$-protocols then generalize it to the quantum state discrimination problem to prove Proposition~\ref{prop:QSD}.  
 
%------------------------------------------------------------------------------------------------------------------------------------------------------%
\section{Semidefinite programming cheating strategy formulations} 
\label{sect:SDP} 
%------------------------------------------------------------------------------------------------------------------------------------------------------% 

In this section, we use the theory of semidefinite programming to formulate Alice and Bob's maximum cheating probabilities for a $\DRIC$-protocol. 
The formulations in this section are a generalization of those for bit-commitment, see~\cite{NST16} and the references therein for details about this special case.   

\subsection{Semidefinite programming}

Semidefinite programming is the theory of optimizing a linear function over a positive semidefinite matrix variable subject to finitely many affine constraints. 
A semidefinite program (SDP) can be written in the following form without loss of generality: 
\begin{equation} \label{primal}
p^* := \sup \{ \inner{C}{X} : \calA(X) = B, \, X \succeq 0 \} 
\end{equation} 
where $\calA$ is a linear transformation, $C$ and $B$ are Hermitian, and $X \succeq Y$ means that $X - Y$ is (Hermitian) positive semidefinite. 

Associated with every SDP is a dual SDP: 
\begin{equation} \label{dual}
d^* := \inf \{ \inner{B}{Y} : \calA^*(Y) = C + S, \, S \succeq 0, \, Y \text{ is Hermitian} \} 
\end{equation} 
where $\calA^*$ is the adjoint of $\calA$. 

We refer to the optimization problem~(\ref{primal}) as the \emph{primal} or \emph{primal SDP} and to the optimization problem~(\ref{dual}) as the \emph{dual} or \emph{dual SDP}. 
We say that the primal is \emph{feasible} if there exists an $X$ satisfying the (primal) constraints  
\[ 
\calA(X) = B 
\quad \text{ and } \quad 
X \succeq 0    
\] 
and we say the dual is \emph{feasible} if there exists $(Y,S)$ satisfying the (dual) constraints  
\[ 
\calA^*(Y) = C + S,  
\quad 
S \succeq 0,     
\quad \text{ and } \quad 
Y \text{ is Hermitian}. 
\] 
If further we have $X$ positive definite, then the primal is said to be \emph{strictly} feasible. If further we have $S$ positive definite, then the dual is said to be \emph{strictly} feasible. 

Semidefinite programming has a rich and powerful duality theory. 
In particular, we use the following:  
\begin{center}
\begin{tabularx}{\textwidth}{rX}
Weak duality: & If the primal and dual are both feasible, then $p^* \leq d^*$. \\ 
Strong duality: & If the primal and dual are both \emph{strictly} feasible, then $p^* = d^*$ 
\emph{and both attain an optimal solution}. \\
\end{tabularx} 
\end{center} 
For more information about semidefinite programming and its duality theory, the reader is referred to~\cite{BV}. 

\subsection{Cheating strategy formulations}

To study a fixed $\DRIC$-protocol, it is convenient to define the following reduced states 
\[ \rho_a := \tr_{\calA} ( \kb{\psi_a} ) \]
for all $a \in [D]$. We show that they appear in both the case of cheating Alice and cheating Bob. 
  
\paragraph{Cheating Bob.}  
To see how Bob can cheat, notice that he only has one message he sends to Alice. 
Thus, he must send $b \in [D]$ to force the outcome he wishes. 
For example, if he wishes to force the outcome $d$, he would send $b$ such that $d = a + b \modd D + 1$. 
Therefore, he must extract the value of $a$ from $\calB$ to accomplish this.  
Suppose he measures $\calB$ with the measurement 
\[ \{ M_1, \ldots, M_D \} \] 
where the outcome of the measurement corresponds to Bob's guess for $a$. 
If Alice chose $a \in [D]$, he succeeds in cheating if his guess is correct, which happens with probability 
\[ \inner{M_a}{\rho_a}. \]  
Since the choice of Alice's integer $a$ is uniformly random, we can calculate Bob's optimal cheating probability as 
\begin{equation} 
\label{Bobprimal}
P_B^* = \max \left\{ \frac{1}{D} \sum_{a \in [D]} \inner{M_a}{\rho_a} : \sum_{a \in [D]} M_a = I_{\calB}, \, M_a \succeq 0, \forall a \in [D] \right\}  
\end{equation} 
noting that the variables being optimized over correspond to a POVM measurement. 
Note that the maximum is attained since the set of feasible $(M_1, \ldots, M_D)$ forms a compact set. 

Now that Bob's optimal cheating probability is stated in terms of an SDP, we can examine its dual as shown in the lemma below. 

\begin{lemma}
For any $\DRIC$-protocol, we have 
\begin{equation} \label{Bobdual} 
P_B^* = \min \left\{ \tr(X) : X \succeq \frac{1}{D} \, \rho_a, \forall a \in [D] \right\}. 
\end{equation} 
\end{lemma} 

\begin{proof} 
One can check using the definitions~(\ref{primal}) and (\ref{dual}) that the optimization problem (\ref{Bobdual}) is the dual of (\ref{Bobprimal}). 
Defining $M_a = \frac{1}{D} I_{\calB}$, for all $a \in [D]$, yields a strictly feasible solution for the primal. Also, $X = I_{\calB}$ is a strictly feasible solution for the dual. 
Thus, by strong duality, both the primal and dual attain an optimal solution and their optimal values are the same.  
\end{proof} 

We refer to the optimization problem~(\ref{Bobprimal}) as \emph{Bob's primal SDP} and to the optimization problem~(\ref{Bobdual}) as \emph{Bob's dual SDP}.  
The utility of having dual SDP formulations is that any feasible solution yields an \emph{upper bound} on the maximum cheating probability. 
Proving upper bounds on cheating probabilities would otherwise be a very hard task. 
 
\paragraph{Cheating Alice.} 
If Alice wishes to force Bob to accept outcome $d \in [D]$, she must convince him that the state in $\calA \otimes \calB$ is indeed $\ket{\psi_a}$ where $a$ is such that $d = a + b \modd D + 1$. 
Note that this choice of $a$ is determined after learning $b$ from Bob, which occurs with uniform probability. 

To quantify the extent to which Alice can cheat, we examine the states Bob has during the protocol. 
We know that Bob measures and accepts $a$ with the measurement operator $\Pi_{a} := \kb{\psi_a}$. 
Let $(a, \calA)$ be Alice's last message. 
Then Bob's state at the end of the protocol is given by a density operator $\sigma_a$ acting on $\calA \otimes \calB$ which is accepted with probability $\inner{\sigma_a}{\kb{\psi_a}}$. 
Note that Alice's first message $\calB$ is in state 
$\sigma := \tr_{\calA}(\sigma_a)$ which is independent of $a$ (since Alice's first message does not depend on $a$ when she cheats). 
Thus, the states under Bob's control are subject to the constraints 
\begin{equation} 
\label{Alice_constraints}
\tr_{\calA}(\sigma_a) = \sigma, \, \forall a \in [D], \quad 
\tr(\sigma) = 1, \quad 
\sigma, \sigma_1 \ldots, \sigma_D \succeq 0. 
\end{equation} 
(Note that $\tr(\sigma_a) = 1$, for all $a \in [D]$, is implied by the constraints above, and is thus omitted.)  
On the other hand, if Alice maintains a purification of the states above, then using Uhlmann's Theorem~\cite{Uhl76} she can prepare any set of states satisfying conditions~(\ref{Alice_constraints}). 

Thus, we have 
\begin{equation} 
\label{Aliceprimal} 
P_A^* = \max \left\{ 
\frac{1}{D} \sum_{a \in [D]} \inner{\sigma_a}{\kb{\psi_a}} 
: 
\tr_{\calA}(\sigma_a) = \sigma, \, \forall a \in [D],  \, 
\tr(\sigma) = 1, \, 
\sigma, \sigma_1 \ldots, \sigma_D \succeq 0   
\right\}. 
\end{equation}  
Again, since the set of feasible $(\sigma, \sigma_1, \ldots, \sigma_D)$ is compact, the above SDP attains an optimal solution. 
 
Similar to the case of cheating Bob, we can view the dual of Alice's cheating SDP above as shown in the lemma below.  
 
\begin{lemma}
For any $\DRIC$-protocol, we have 
\begin{equation} \label{Alicedual} 
P_A^* 
=  
\min \left\{ s : s I_{\calB} \succeq \sum_{a \in [D]} Z_a, \, 
I_{\calA} \otimes Z_a \succeq \frac{1}{D} \kb{\psi_a}, \, \forall a \in [D], \, Z_a \textup{ is Hermitian} 
\right\}.  
\end{equation} 
\end{lemma} 
  
\begin{proof} 
It can be checked that (\ref{Alicedual}) is in fact the dual of (\ref{Aliceprimal}). 
By defining $\sigma$ and each $\sigma_1, \ldots, \sigma_D$ to be completely mixed states, we have that the primal is strictly feasible. 
By defining $s = D+1$ and each $Z_1, \ldots, Z_D$ to be equal to $I_{\calB}$, we have that the dual is strictly feasible as well. The result now holds by applying strong duality.  
\end{proof}  

We refer to the optimization problem~(\ref{Aliceprimal}) as \emph{Alice's primal SDP} and the optimization problem~(\ref{Alicedual}) as \emph{Alice's dual SDP}. 
 
Note that every solution feasible in Alice's dual SDP has $Z_a$ being positive semidefinite, for all $a \in [D]$. 
We can further assume that each $Z_a$ is positive definite if we sacrifice the attainment of an optimal solution. 
This is because we can take an optimal solution $(s, Z_1, \ldots, Z_D)$ and consider ${(s + \eps D, Z_1 + \eps I_{\calB}, \ldots, Z_D + \eps I_{\calB})}$ which is also feasible for any $\eps > 0$, and $s + \eps D$ approaches $s = P_A^*$ as $\eps$ decreases to $0$. 
  
Next, we use an analysis similar to one found in~\cite{Moc05} and \cite{Wat09} to simplify the constraint 
$I_{\calA} \otimes Z_a \succeq \kb{\psi_a}$ 
when $Z_a$ is positive definite. 
Since $X \to ZXZ^{-1}$ is an automorphism of the set of positive semidefinite matrices for any fixed positive definite $Z$, we have  
\begin{equation} 
I_{\calA} \otimes Z_a \succeq \frac{1}{D} \kb{\psi_a} 
\iff  
I_{\calA \otimes \calB} \succeq (I_{\calA} \otimes Z_a^{-1/2}) \left( \frac{1}{D} \kb{\psi_a} \right) (I_{\calA} \otimes Z_a^{-1/2}). 
\label{eqn:trace}
\end{equation} 
Note that since the quantity on the right is positive semidefinite with rank at most $1$, its largest eigenvalue is equal to its trace which is equal to 
\[  
\frac{1}{D} \inner{I_{\calA} \otimes Z_a^{-1} } {\kb{\psi_a}} = \frac{1}{D} \inner{Z_a^{-1}}{\tr_{\calA}(\kb{\psi_a})}
= \frac{1}{D} \inner{Z_a^{-1}}{\rho_a}. 
\]
Thus, we can rewrite~(\ref{eqn:trace}) as 
\[ 
I_{\calA} \otimes Z_a \succeq \frac{1}{D} \kb{\psi_a} 
\iff  
\frac{1}{D} \inner{Z_a^{-1}}{\rho_a} \leq 1 
\iff
\inner{Z_a^{-1}}{\rho_a} \leq D. 
\] 
Therefore, we have the following lemma. 

\begin{lemma}
For any $\DRIC$-protocol, we have 
\begin{equation} \label{eqn:newdual}  
P_A^* 
= 
\inf 
\left\{ 
s 
: 
s I_{\calB} \succeq \sum_{a \in [D]} Z_a, \, 
\inner{Z_a^{-1}}{\rho_a} \leq D, \forall a \in [D], \, 
Z_a \textup{ is positive definite}, \, \forall a \in [D] 
\right\}. 
\end{equation} 
\end{lemma} 

We also refer to the optimization problem~(\ref{eqn:newdual}) as Alice's dual SDP and we distinguish them by equation number. 
 
%------------------------------------------------------------------------------------------------------------------------------------------------------%
\section{Finding a decent $\DRIC$-protocol} 
\label{sect:qprotocol}
%------------------------------------------------------------------------------------------------------------------------------------------------------% 

In this section, we exhibit a $\DRIC$-protocol which has the same cheating probabilities as Protocol~\ref{protocol:classical}:  
\[ P_B^* = \frac{1}{m} \quad \text{ and } \quad P_A^* = \frac{m}{D}. \] 
To do this, define $T_m$ to be the subsets of $[D]$ of cardinality $m$ and note that $|T_m| = \chews{D}{m}$. 
Consider the following states 
\[ \ket{\psi_a} := {\dfrac{1}{\sqrt{\chews{D-1}{m-1}}}}\dsum_{S \in T_m \ : \ a \in S} \ket{S}\ket{S} \in \calA \otimes \calB, \] 
for $a \in [D]$, where $\calA = \calB = \C^{|T_m|}$. Notice that  
\[ \rho_a := \tr_{\calA} \left( \kb{\psi_a} \right) = {\dfrac{1}{{\chews{D-1}{m-1}}}}\sum_{S \in T_m \ : \ a \in S} \kb{S}. \]

We now use the cheating SDPs developed in the previous section to analyze the cheating probabilities of this protocol. 

\paragraph{Cheating Bob.} 
To prove that Bob can cheat with probability at least $1/m$, suppose he measures his message from Alice in the computational basis. 
He then obtains a random subset $S \in T_m$ such that $a \in S$. 
He then guesses which integer is $a$ and responds with the appropriate choice for $b$ to get his desired outcome. He succeeds if and only if his guess for $a$ (from the $m$ choices in $S$) is correct. This strategy succeeds with probability $1/m$. Thus,  $P_B^* \geq 1/m$.

To prove Bob cannot cheat with probability greater than $1/m$, notice that $X = \dfrac{1}{D \chews{D-1}{m-1}} I_{\calB}$ satisfies 
\[  
X \succeq \frac{1}{D} \rho_a, \; \forall a \in [D],  
\] 
and thus is feasible in Bob's dual (\ref{Bobdual}). Therefore, $P_B^* \leq \tr(X) = 1/m$, as desired.  

\paragraph{Cheating Alice.} 
Alice can cheat by creating the maximally entangled state 
\[ \ket{T_m} := \frac{1}{\sqrt{|T_m|}} \sum_{S \in T_m} \ket{S}\ket{S} \in \calA \otimes \calB \] 
and sending $\calB$ to Bob. 
After learning $b$, she sends $a$ such that $a + b \modd D + 1$ is her desired outcome. 
She also sends $\calA$ to Bob (without altering it in any way).  
Thus, her cheating probability is precisely the probability of her passing Bob's cheat detection which is 
\[ \inner{\Pi_a}{\kb{T_m}} = \inner{\kb{\psi_a}}{\kb{T_m}} = |\brakett{T_m}{\psi_a}|^2 = \frac{m}{D}. \] 
Therefore, this cheating strategy succeeds with probability $m/D$, proving $P_A^* \geq m/D$. 

To prove this strategy is optimal, we use Alice's dual~(\ref{eqn:newdual}). 
Define 
\[ 
Z_a 
:= 
\dfrac{1}{D} \sum_{S \in T_m \ : \ a \in S} \kb{S} 
+ 
\eps \sum_{S \in T_m \ : \ a \not\in S} \kb{S} 
\] 
where $\eps$ is a small positive constant. $Z_a$ is invertible and we can write 
\[ 
Z_a^{-1} 
:= 
D \sum_{S \in T_m \ : \ a \in S} \kb{S} 
+ 
\dfrac{1}{\eps} \sum_{S \in T_m \ : \ a \not\in S} \kb{S}. 
\] 
We see that each $Z_a$ satisfies $\inner{Z_a^{-1}}{\rho_a} = D$, for all $a \in [D]$. Also, 
\[ Z_a 
\preceq 
\dfrac{1}{D} \sum_{S \in T_m \ : \ a \in S} \kb{S} 
+ \eps I_{\calB}
\] 
thus 
\[ \sum_{a \in [D]} Z_a 
\preceq 
\dfrac{1}{D} \sum_{a \in [D]} \sum_{S \in T_m \ : \ a \in S} \kb{S} 
+ \eps \, D \, I_{\calB} 
= \left( \dfrac{m}{D} + \eps D \right) \, I_{\calB}.  
\] 
Thus, $s = \dfrac{m}{D} + \eps D$ satisfies 
\[ s \, I_{\calB} \succeq \sum_{a \in [D]} Z_a \] 
proving $P_A^* \leq s = \dfrac{m}{D} + \eps D$, for all $\eps > 0$. Therefore, $P_A^* = m/D$, as desired. 

%------------------------------------------------------------------------------------------------------------------------------------------------------%
\section{Balancing Alice and Bob's cheating probabilities} 
\label{sect:balancing} 

This section is comprised of two parts. We first focus on reducing Bob's cheating probabilities, then Alice's. 

\subsection{Building new protocols that reduce Bob's cheating} 
%------------------------------------------------------------------------------------------------------------------------------------------------------%

We start with a lemma. 

\begin{lemma} \label{lem:Bob}
If there exists a $\DRIC$-protocol with cheating probabilities $P_A^* = \alpha$ and $P_B^* = \beta$, then there exists another $\DRIC$-protocol with cheating probabilities $P_A^* = \alpha'$ and $P_B^* = \beta'$ where 
\[ \beta' \leq (1-t) \beta + \frac{t}{D} \quad \text{ and } \quad \alpha' \leq (1-t) \alpha + t. \] 
for any $t \in (0,1)$. 
\end{lemma} 

\begin{proof} 
To prove this lemma, fix a $\DRIC$-protocol with cheating probabilities $P_A^* = \alpha$ and $P_B^* = \beta$ defined by the states $\ket{\psi_a} \in \calA \otimes \calB$, for $a \in [D]$. 
Extend each of the  Hilbert spaces $\calA$ and $\calB$ by another basis vector $\ket{\perp}$ and denote these Hilbert spaces by $\calA'$ and $\calB'$, respectively. 
In short, $\calA' := \calA \oplus \spann \{ \ket{\perp} \}$ and $\calB' := \calB \oplus \spann \{ \ket{\perp} \}$. 
Note that 
\[ 
\brakett{\perp, \perp \!}{\psi_a} = 0, \; \text{ for all } a \in [D]. 
\] 

We now analyze the cheating probabilities of Alice and Bob in the new $\DRIC$-protocol defined by the states 
\[ 
\ket{\psi'_a} := \sqrt{1-t} \ket{\psi_a} + \sqrt{t} \ket{\perp, \perp} \in \calA' \otimes \calB', \; \text{ for all } a \in [D] 
\] 
as a function of $t \in (0,1)$.  
For this, note that 
\[ 
\rho'_a := \tr_{\calA} \left( \kb{\psi'_a} \right) = (1-t) \, \rho_a + t \, \kb{\perp},  
\]   
where $\rho_a := \tr_{\calA} \left( \kb{\psi_a} \right)$. 

\medskip
Intuitively, Alice can cheat more if the states $\ket{\psi_a}$ are ``close'' to each other. On the other hand, Bob can cheat more if they are ``far apart'', at least on the subsystem he receives for the first message. What this protocol modification does is make all the states closer together to increase Alice's cheating probability but to decrease Bob's. 

\paragraph{Cheating Bob.} 
Let $X$ be an optimal solution to Bob's dual~(\ref{Bobdual}) for the original protocol. 
So $\tr(X) = \beta$ and $X \succeq \frac{1}{D} \, \rho_a$, for all $a \in [D]$. 

To upper bound Bob's cheating probability in the new protocol, we show that 
\[ 
X' := (1-t)X + \frac{t}{D} \kb{\perp} 
\] 
is feasible for Bob's dual for the new protocol.   
We have 
\[ X' = (1-t)X + \frac{t}{D} \kb{\perp} \succeq \frac{1-t}{D} \rho_a + \frac{t}{D}  \kb{\perp} = \frac{1}{D} \rho'_a, \] 
for all $a \in [D]$. 
Thus $X'$ is feasible proving that $P_B^* \leq \tr(X') = (1-t) \beta + t/D$ for the new protocol. 
 
\paragraph{Cheating Alice.} 
We now repeat the same process for Alice. 
Let $(s, Z_1, \ldots, Z_D)$ be a feasible solution for Alice's dual~(\ref{eqn:newdual}) for the original protocol. 
That is, $s I_{\calB} \succeq \sum_{a \in [D]} Z_a$ and  each positive definite $Z_a$ satisfies $\inner{Z_a^{-1}}{\rho_a} \leq D$, for each $a \in [D]$. 
Define 
\[  
Z'_a := \delta \, Z_a + \eps \, \kb{\perp}, 
\] 
for $a \in [D]$, for some choice of positive constants $\delta$ and $\eps$ to be specified later. 
Notice that 
\[ (Z'_a)^{-1} = \dfrac{1}{\delta} \, Z_a^{-1} + \dfrac{1}{\eps} \kb{\perp}. \]   
To show the analogous constraints are satisfied with $Z'_a$, recall that $\inner{\kb{\perp}}{\rho_a} = 0$ for all $a \in [D]$.  
Using this, we have 
\[ 
\inner{(Z_a')^{-1}}{\rho'_a} 
=  
\dfrac{1}{\delta} \inner{Z_a^{-1}}{\rho'_a} 
+ 
\dfrac{1}{\eps} \inner{\kb{\perp}}{\rho'_a} 
\leq  
\dfrac{D(1-t)}{\delta} + \dfrac{t}{\eps}. 
\] 
By choosing $\delta$ and $\eps$ appropriately, we can make the quantity on the right equal to $D$. 
To finish the proof of feasibility, note that 
\[ 
\sum_{a \in [D]} Z'_a 
=
\delta \sum_{a \in [D]} Z_a + \eps D \kb{\perp}
\preceq 
\delta s \, I_{\calB} + \eps D \kb{\perp} 
\preceq 
s' I_{\calB'} 
\] 
where $s' := \max \{ \delta s, \eps D \}$. 
By choosing 
\[ 
\eps =  \dfrac{s (1-t) + t}{D} > 0
\quad \text{ and } \quad 
\delta = (1 - t) + \frac{t}{s} > 0
\]
we get $\inner{(Z'_a)^{-1}}{\rho'_a} \leq D$ and $s' = s (1-t) + t$. Since $s$ can be taken to be arbitrarily close to $\alpha$, we have $P_A^* \leq (\alpha + \eps')(1-t) + t$ for all $\eps' > 0$, finishing the proof. 
\end{proof} 

Note that this lemma is useful when $\beta > \alpha$. 
In this case, one can choose 
\[ t = \dfrac{\beta - \alpha}{(1 - 1/D) + (\beta - \alpha)} \in (0,1) \] 
to equate the upper bounds. 
If $\alpha > \beta$, then no choice of $t \in (0,1)$ will make the two upper bounds in Lemma~\ref{lem:Bob} equal. 
We summarize in the following corollary. 

\begin{corollary} \label{cor:Bob}
If there exists a $\DRIC$-protocol with cheating probabilities $P_A^* = \alpha$ and $P_B^* = \beta$, \emph{with $\beta > \alpha$}, then there exists another $\DRIC$-protocol with maximum cheating probability 
\[ \max \{ P_A^*, P_B^* \} \leq \frac{D \beta - \alpha}{D \beta - D \alpha + D - 1} < \beta. \] 
\end{corollary} 
 
%------------------------------------------------------------------------------------------------------------------------------------------------------%
\subsection{Building new protocols that reduce Alice's cheating}    
%------------------------------------------------------------------------------------------------------------------------------------------------------%

In this subsection, we show how to reduce Alice's cheating probabilities in a $\DRIC$-protocol.  

\begin{lemma} \label{lem:Alice}
If there exists a $\DRIC$-protocol with cheating probabilities $P_A^* = \alpha$ and $P_B^* = \beta$, then there exists another $\DRIC$-protocol with cheating probabilities $P_A^* = \alpha'$ and $P_B^* = \beta'$ where 
\[ \beta' \leq (1-t) \beta + t \quad \text{ and } \quad \alpha' \leq (1-t) \alpha + \frac{t}{D}, \] 
for $t \in (0,1)$. 
\end{lemma}  

\begin{proof} 
To prove this lemma, fix a $\DRIC$-protocol with cheating probabilities $P_A^* = \alpha$ and $P_B^* = \beta$ defined by the states $\ket{\psi_a} \in \calA \otimes \calB$, for $a \in [D]$. 
Extend each of the Hilbert spaces $\calA$ and $\calB$ by the set of orthogonal basis vectors 
$\{ \ket{\perp_a} : a \in [D] \}$, and denote these new Hilbert spaces by $\calA'$ and $\calB'$, respectively. 
In other words, 
\[ 
\calA' := \calA \oplus \spann \{ \ket{\perp_1}, \ldots, \ket{\perp_D} \} 
\quad \text{ and } \quad 
\calB' := \calB \oplus \spann \{ \ket{\perp_1}, \ldots, \ket{\perp_D} \}. 
\] 
Note that 
\[ \brakett{\perp_{a''}, \perp_{a'} \!}{\psi_{a}} = 0, \; \text{ for all } a, a', a'' \in [D]. \] 
Again, we analyze the cheating probabilities of Alice and Bob in the new $\DRIC$-protocol defined by the states 
\[ \ket{\psi'_a} := \sqrt{1-t} \ket{\psi_a} +  \sqrt{t} \ket{\perp_a} \ket{\perp_a} \in \calA' \otimes \calB'  
\] 
for $a \in [D]$. The reduced states are 
\[ \rho'_a := (1-t) \, \rho_a + t \, \kb{\perp_a} 
\] 
for $a \in [D]$, recalling that $\rho_a := \tr_{\calA}(\kb{\psi_a})$. 
We now analyze the cheating probabilities of this new protocol as a function of $t \in (0,1)$. 

\medskip 
Intuitively, this protocol modification works in the opposite manner of the last. Here, we are making the states farther apart as to decrease Alice's cheating at the expense of increasing Bob's. 

\paragraph{Cheating Bob.} 
Let $X$ be an optimal solution for Bob's dual~(\ref{Bobdual}) for the original protocol. 
Define  
\[ X' := (1-t) X + \frac{t}{D} \sum_{a \in [D]} \kb{\perp_a} \] 
which can easily be seen to be feasible in the dual SDP for the new protocol. Thus, we have  $P_B^* \leq \tr(X') = (1-t) \beta + t$.  
  
\paragraph{Cheating Alice.} 
Let $(s, Z_1, \ldots, Z_D)$ be a feasible solution for Alice's dual~(\ref{eqn:newdual}) for the original protocol. 
That is, $s I_{\calB} \succeq \sum_{a \in [D]} Z_a$ and  each positive definite $Z_a$ satisfies $\inner{Z_a^{-1}}{\rho_a} \leq D$, for each $a \in [D]$. 

Define 
\[ 
Z'_a 
:= 
\delta \, Z_a 
+ 
\eps \, \kb{\perp_a} 
+ 
\zeta \, \sum_{c \in [D], c \neq a} \kb{\perp_c} \] 
for positive constants $\delta, \eps, \zeta$ to be specified later.  
Note that $\inner{\sum_{c \in [D], c \neq a} \kb{\perp_c}}{\rho'_a} = 0$, for all $a \in [D]$. 
We have $Z'_a$ is invertible and we can write its inverse as  
\[ 
(Z'_a)^{-1} 
= 
\frac{1}{\delta} Z_a^{-1} 
+ 
\frac{1}{\eps} \kb{\perp_a} 
+ 
\frac{1}{\zeta} \, \sum_{c \in [D], c \neq a} \kb{\perp_c}   
\]  
which satisfies 
\[ 
\inner{(Z'_a)^{-1}}{\rho'_a} 
= 
\frac{1}{\delta} 
\inner{Z_a^{-1}}{\rho'_a} 
+ 
\frac{1}{\eps} 
\inner{\kb{\perp_a}}{\rho'_a} 
\leq 
\frac{D (1-t)}{\delta} + \frac{t}{\eps}. 
\] 
Also note that
\begin{eqnarray*}
\sum_{a \in [D]} Z'_a 
& = & 
\delta \, \sum_{a \in [D]} Z_a 
+ 
\eps \, \sum_{a \in [D]} \kb{\perp_a} 
+ 
\zeta \, \sum_{a \in [D]} \sum_{c \in [D], c \neq a} \kb{\perp_c} \\
& \preceq & 
\delta s \, I_{\calB} + (\eps + \zeta(D-1)) \, \sum_{a \in [D]} \kb{\perp_a} \\ 
& \preceq & 
s' I_{\calB'}
\end{eqnarray*} 
where $s' := \max \{ \delta s, \eps + \zeta(D-1)  \}$.  
Setting  
\[ 
\eps = (1-t)s + \frac{t}{D} > 0 
\quad \text{ and } \quad 
\delta = (1-t) + \frac{t}{Ds} > 0 
\] 
we get $\inner{(Z'_a)^{-1}}{\rho'_a} \leq D$ and $s' = (1-t)s + t/D + \zeta(D-1)$. 
Since $s$ can be taken to be arbitrarily close to $\alpha$, and $\zeta$ arbitrarily close to $0$, we have $P_A^* \leq (\alpha + \eps')(1-t) + t/D + \eps'(D-1)$ for all $\eps' > 0$, 
finishing the proof. 
\end{proof} 
 
As opposed to Lemma~\ref{lem:Bob}, the above lemma is useful when $\alpha > \beta$.  
Similarly, if $\beta > \alpha$, then no choice of $t \in (0,1)$ will make the two upper bounds in Lemma~\ref{lem:Alice} equal. 

By symmetry, we have the following corollary. 

\begin{corollary} \label{cor:Alice}
If there exists a $\DRIC$-protocol with cheating probabilities $P_A^* = \alpha$ and $P_B^* = \beta$, \emph{with $\alpha > \beta$}, then there exists another $\DRIC$-protocol with maximum cheating probability 
\[ \max \{ P_A^*, P_B^* \} \leq \frac{D \alpha - \beta}{D \alpha - D \beta + D - 1} < \alpha. \] 
\end{corollary} 

Note that if $\alpha = \beta$, the quantity $\frac{D \alpha - \beta}{D \alpha - D \beta + D - 1}$ is equal to $\alpha (= \beta)$. Thus, we still have 
\[ \max \{ P_A^*, P_B^* \} \leq \frac{D \alpha - \beta}{D \alpha - D \beta + D - 1} \] 
holding, although no protocol modification is necessary. 
Therefore, Proposition~\ref{prop:balancing} now follows from combining Corollaries~\ref{cor:Bob} and \ref{cor:Alice} and the comment above.  
   
\section{Kitaev's lower bound and quantum state discrimination} 
\label{sect:Kitaev} 

We start this section with a short proof of Kitaev's lower bound for $\DRIC$-protocols. 

\subsection{Kitaev's lower bound}

Let $(s, Z_1, \ldots, Z_D)$ be an optimal solution for Alice's dual SDP~(\ref{Alicedual}), i.e., 
\[
P_A^* = s, 
\quad 
s I_{\calB} \succeq \sum_{a \in [D]} Z_a, 
\quad 
\text{ and } 
\quad 
I_{\calA} \otimes Z_a \succeq \frac{1}{D} \kb{\psi_a}, 
\, \text{ for all } \, 
a \in [D]. 
\]  
Note that from the last constraint in the SDP, we require that $Z_a$ is positive semidefinite for all $a \in [D]$. 
We may assume that $s I_{\calB} = \sum_{a = 1}^D Z_a$, without loss of generality, since we can always increase $Z_1$ to make this the case. I.e., we can redefine $Z_1 \to Z_1 + \left(  s I_{\calB} - \sum_{a \in [D]} Z_a \right)$ 
which maintains the same value for $s$ and still satisfies all the constraints. 
Define the matrices $M_a := \frac{1}{s} Z_a$ for all $a \in [D]$.
We see this is feasible for Bob's cheating SDP~(\ref{Bobprimal}). 
We thus have that 
\[ 
P_B^* 
 \geq  
\frac{1}{D} \sum_{a = 1}^D \inner{\rho_a}{M_a}  
=  
\frac{1}{sD} \sum_{a = 1}^D \inner{\rho_a}{Z_a}  
=  
\frac{1}{sD} \sum_{a = 1}^D \inner{\kb{\psi_a}}{I_{\calA} \otimes Z_a} 
\geq  
\frac{1}{sD^2} \sum_{a = 1}^D \inner{\kb{\psi_a}}{\kb{\psi_a}} 
\]
implying that  
$P_A^* P_B^* \geq {1}/{D}$, which is precisely Kitaev's lower bound for die-rolling. 

\begin{remark} 
This proof is slightly different than Kitaev's original proof which involves combining  Bob's and Alice's optimal dual solutions. 
The above proof takes an optimal dual solution for Alice, then creates a valid cheating strategy for Bob.  
This new perspective could shed light on the nature of dual solutions and their role in  creating point games (which are still regarded as being quite mysterious). 
Point games are beyond the scope of this work; the interested reader is referred to~\textup{\cite{Moc07, ACGKM15, NST15}} for further details. 
\end{remark} 

\subsection{Quantum state discrimination} 

Consider now a $\DRIC$-protocol but Alice now chooses $a \in [D]$ with probably $p_a$ (instead of uniformly at random). 
Then the amount Bob can cheat in this modified protocol exactly  corresponds to the success probability of a quantum state discrimination (QSD) problem. 
 
We can easily modify the optimization problem~(\ref{Bobprimal}) to see that the optimal success probability in the QSD problem is given by 
\begin{equation*} 
\beta := \label{QSD}
\max \left\{ \sum_{a \in [D]} p_a \inner{M_a}{\rho_a} : \sum_{a \in [D]} M_a = I_{\calB}, \, M_a \succeq 0, \forall a \in [D] \right\},  
\end{equation*} 
where we denote the optimal value as $\beta$ (to distinguish its context from cryptographic security for the moment). 
 
Consider again Alice's dual SDP~(\ref{Alicedual}) 
\[
\alpha := \min \left\{ s : s I_{\calB} \succeq \sum_{a \in [D]} Z_a, \, 
I_{\calA} \otimes Z_a \succeq \frac{1}{D} \kb{\psi_a}, \, \forall a \in [D], \, Z_a \textup{ is Hermitian} 
\right\}.
\]   
Then repeating the proof of Kitaev's lower bound above, 
we get that $\beta \, \alpha \geq {1}/{D}$. 
We now bound $\beta$ by bounding $\alpha$:  
\begin{eqnarray*} 
\alpha 
& = & 
\min \left\{ s : s I_{\calB} \succeq \sum_{a \in [D]} Z_a, \, 
I_{\calA} \otimes Z_a \succeq \frac{1}{D} \kb{\psi_a}, \, \forall a \in [D], \, Z_a \textup{ is Hermitian} 
\right\} \\ 
& = & 
\inf 
\left\{ 
s : s I_{\calB} \succeq \sum_{a \in [D]} Z_a, \, 
\inner{Z_a^{-1}}{\rho_a} \leq D, \forall a \in [D], \, 
Z_a \textup{ is positive definite}, \, \forall a \in [D] 
\right\} \\ 
& = & 
\inf 
\left\{ 
\lambda_{\max} \left( \sum_{a \in [D]} Z_a \right) :  
\inner{Z_a^{-1}}{\rho_a} \leq D, \forall a \in [D], \, 
Z_a \textup{ is positive definite}, \, \forall a \in [D] 
\right\}
\end{eqnarray*} 
where $\lambda_{\max}$ denotes the largest eigenvalue of a Hermitian matrix. Since $\lambda_{\max}(A) = ( \lambda_{\min}(A^{-1}) )^{-1}$ for $A$ positive definite, we have 
\[ 
\alpha 
=
\left( 
\sup 
\left\{ 
\lambda_{\min} \left( \left( \sum_{a \in [D]} Z_a \right)^{-1} \right) :  
\inner{Z_a^{-1}}{\rho_a} \leq D, \forall a \in [D], \, 
Z_a \textup{ is positive definite}, \, \forall a \in [D] 
\right\} \right)^{-1}. 
\]  
It now follows easily that 
\[ 
\frac{1}{\alpha D}  
=
\sup 
\left\{ 
\lambda_{\min} \left( \left( \sum_{a \in [D]} (D \, Z_a) \right)^{-1} \right) :  
\inner{Z_a^{-1}}{\rho_a} \leq D, \forall a \in [D], \, 
Z_a \textup{ is positive definite}, \, \forall a \in [D] 
\right\} . 
\]  
Now Proposition~\ref{prop:QSD} follows by defining $W_a := (D Z_a)^{-1}$ for all $a \in [D]$. 
 
We now mention how Proposition~\ref{prop:QSD} can be tight. We see that if we view the QSD problem from the perspective of a cheating Bob in a $\DRIC$-protocol, then the (non)tightness of Proposition~\ref{prop:QSD} is exactly characterized by the (non)tightness of Kitaev's lower bound above. Thus, the examples of $\DRIC$-protocols saturating Kitaev's lower bound, i.e., $P_B^* P_A^* = 1/D$, yield instances of the QSD problem where Proposition~\ref{prop:QSD} is tight.  
 
%%%%%%%%%%%%%%%%%%%%%%%%%%%%%%%%%%%%%%%%%%%%%%%%%%%%%%%%
 
\section*{Acknowledgements}
I thank Sevag Gharibian for useful discussions. 
J.S. is supported in part by NSERC Canada. 

Research at the Centre for Quantum Technologies at the National University of Singapore is partially funded by the Singapore Ministry of Education and the National Research Foundation, also through the Tier 3 Grant ``Random numbers from quantum processes,'' (MOE2012-T3-1-009). 

\nocite{CK11}
\nocite{CK09}
\nocite{Kit03} 
\nocite{CKS13}
\nocite{NST15}
\nocite{NST16}
\nocite{AS10}
\nocite{LC97}
\nocite{May97}
\nocite{LC97a}
\nocite{NS03}
\nocite{SR01} 
    
\bibliographystyle{alpha}
\bibliography{CCF_bib.bib} 

\end{document}